\documentclass[a4paper]{jpconf} 
\usepackage{color} 
\usepackage{framed}

\usepackage{amsmath,amssymb,amsthm,array,booktabs,dsfont,graphicx,mathtools,multirow,rotating,stmaryrd,subfig,tensor,verbatim}
\usepackage[bookmarksnumbered,linktocpage,pdfstartview=FitH]{hyperref}
\usepackage[all]{hypcap}

\newcolumntype{M}[1]{>{$}{#1}<{$}}
\newcolumntype{T}[1]{>{\bgroup\tiny $}{#1}<{$\egroup}}
\newcolumntype{K}{>{\lvert}{c}<{\rangle}}
\newcolumntype{S}[1]{>{\centering\arraybackslash} m{#1cm}}
\newcolumntype{D}[1]{>{$\displaystyle}{#1}<{$}}

\DeclarePairedDelimiter{\ket}{\lvert}{\rangle}

\newcommand{\rep}[1]{\ensuremath{\mathbf{#1}}}

\newcolumntype{B}[1]{>{\mathbf\bgroup}{#1}<{\egroup}}
\newcolumntype{K}{>{\lvert}{c}<{\rangle}}

\newtheorem{theorem}{Theorem}
\newtheorem{lemma}{Lemma}

\newtheorem{definition}[theorem]{Definition}

\numberwithin{equation}{section}

\DeclareMathOperator{\Aut}{Aut}
\DeclareMathOperator{\Str}{Str}
\DeclareMathOperator{\Det}{Det}

\DeclareMathOperator{\Iso}{Iso}

\DeclareMathOperator{\SO}{SO}

\DeclareMathOperator{\SL}{SL}

\newcommand{\pmtwo}[4]{\begin{pmatrix}#1 & #2 \\ #3 & #4 \end{pmatrix}}
\newcommand{\pmthree}[9]{\begin{pmatrix}#1 & #2&#3 \\ #4 & #5&#6 \\#7 & #8&#9  \end{pmatrix}}
\newcommand{\be}{\begin{equation}}
\newcommand{\ee}{\end{equation}}
\newcommand{\bea}{\begin{eqnarray}}
\newcommand{\eea}{\end{eqnarray}}

\newcommand{\half}{\tfrac{1}{2}}

\newcommand{\J}{\mathfrak{J}}

\newcommand{\F}{\mathds{F}}

\newcommand{\C}{\mathds{C}}

\newcommand{\Oct}{\mathds{O}}

\newcommand{\FTS}{\mathfrak{F}}

\renewcommand{\thefootnote}{\fnsymbol{footnote}}

\begin{document}

\title{Qubit entanglement: A Jordan algebraic perspective$^\star$}

\author{L. Borsten}
 \address{Blackett Laboratory, Imperial College, London, SW7 2AZ, U.K.}
\ead{leron.borsten@imperial.ac.uk}
\begin{abstract}
We review work classifying the physically distinct forms of 3-qubit entanglement using the elegant framework of Jordan algebras, Freudenthal triple systems and groups of type $E_7$.  While this framework is, in the first instance, specific to three qubits, it is shown here how the essential features may be naturally generalised to an arbitrary number of qubits. 

$^\star$Based on a talk given at \emph{3Quantum: Algebra Geometry Information} (AGMP Network), Tallinn University of Technology, Estonia, 10-13 July 2012. 
\end{abstract}


\section{Introduction}

This contribution relates a talk given at \emph{3Quantum} (2012) in the sub-theme on Jordan algebras. It also touched on two further naively disconnected  sub-themes, quantum information and string theory, and so in some limited sense satisfied  the interdisciplinary aspirations  of \emph{3Quantum}. However, given the length constraints, for the sake  of clarity we focus here on the interplay between Jordan algebras and quantum entanglement. The interested reader may consult \cite{Borsten:2008wd, Borsten:2010db, Borsten:2012fx} and the references therein for the unreported connections to black holes in string theory.

The central idea is that the Freudenthal triple system (FTS) based on the smallest  spin-factor Jordan algebra naturally captures the classification of the physically distinct forms of 3-qubit entanglement as  prescribed by the paradigm of \emph{stochastic local operations and classical communication} (SLOCC) \cite{Borsten:2009yb}. In particular, it is shown that the four  \emph{Freudenthal ranks} correspond precisely to the four 3-qubit entanglement classes: (1) Totally separable $A$-$B$-$C$, (2) Biseparable $A$-$BC$, $B$-$CA$, $C$-$AB$, (3) Totally entangled $W$, (4) Totally entangled GHZ (Greenberger-Horne-Zeilinger). This agrees perfectly with the  SLOCC classification first obtained  in \cite{Dur:2000}. The rank 4 GHZ class is regarded as maximally entangled in the sense that it has non-vanishing quartic norm, the defining invariant of the Freudenthal triple system.  

Our hope is give an account of this idea which is accessible to both the quantum information  and Jordan algebra communities.  Accordingly, we begin in \autoref{sec:entang} with an elementary introduction to the essential concepts of entanglement and entanglement classification. In particular, a three player non-local game is used to illustrate that three qubits can be totally entangled in two physically  distinguishable  ways and  how the concept of SLOCC is used to classify different forms of entanglement. Then, in \autoref{sec:3qubitfts} we briskly review the necessary elements of cubic Jordan algebras and the FTS before defining the \emph{3-qubit Jordan algebra} and the corresponding \emph{3-qubit FTS}. Having laid the foundations we proceed in \autoref{sec:FTSclass} to demonstrate how the 3-qubit FTS elegantly characterises the 3-qubit entanglement classes via the FTS \emph{ranks}.  Finally, in \autoref{sec:nquibitFTS}, we illustrate how the key features of the 3-qubit FTS may be generalised to an arbitrary number of qubits.

\section{How entangled is totally entangled?}\label{sec:entang}

The importance of entanglement to quantum computing protocols has precipitated a need to characterise, qualitatively and quantitatively, the ``amount'' of entanglement contained in given state. Indeed, what is actually meant by ``amount'' is somewhat ambiguous and there are many possible measures present in the literature. See, for example, \cite{Plenio:2007, Horodecki:2007}. One might ask, for instance, when can two totally entangled, but ostensibly different,  states  facilitate the same set of non--local quantum protocols. Two states equivalent is this sense could reasonably be considered to have the same degree of entanglement.

Already for pure 3-qubit states this question is non-trivial. Three qubits can be totally entangled in two physically distinct ways \cite{Dur:2000}. These two entanglement forms are represented by the states,
\be\label{eq:ghz}
\ket{W}=\frac{1}{\sqrt{3}}\left(\ket{001}+\ket{010}+\ket{100}\right)\qquad\text{and}\qquad\ket{GHZ}=\frac{1}{\sqrt{2}}\left(\ket{000}+\ket{111}\right).
\ee

To those unfamiliar with entanglement classification, what separates these states  is perhaps  not  entirely obvious; they are both totally entangled and permutation symmetric. Nonetheless, they exhibit distinct non-local/information-theoretic properties. One such difference is nicely drawn out  by the three-player non-local game introduced in \cite{Watrous:2006}. This is essentially a rephrasing  of Mermin's  elegant presentation \cite{Mermin:1990} of a ``Bell-type theorem without inequalities'', the first of which was found by Greenberger, Horne and Zeilinger (GHZ) \cite{Greenberger:1989}. The formulation as a non-local game is, however, especially appealing in that it emphasises  not only the remarkable  non-local properties of these states, but also how they can be used in an information-theoretic sense. More specifically, in the context
of cooperative games of incomplete information  \cite{1313847}.

A non-local game \cite{1313847, Watrous:2006} consists of \emph{players} (Alice, Bob, Charlie\ldots), who act cooperatively in order to win, and  a \emph{referee} who coordinates the game. The players may collectively decide on a strategy before the game commences. Once it has begun they may no longer communicate. Whether or not the players win is determined by the referee. To begin the referee randomly selects one question, from a known fixed set $\mathcal{Q}$, to be sent to each player. The players know only their own questions. Each player must then send back a response from the set of answers $\mathcal{A}$. The referee determines whether the players win using the set of sent questions and received answers according to some predetermined rules. These rules are known to the players before the game gets under way so that they may attempt to devise a winning strategy.   

For the  three-player game \cite{Watrous:2006} the questions sent to Alice, Bob and Charlie, denoted respectively by $r,s$ and $t$,   are taken  from the set $\mathcal{Q}=\{0, 1\}$. However, the referee ensures that $rst\in\{000, 110, 101, 011\}$ with a uniform distribution and the players are aware of this. The answers $a, b, c$,  sent back by Alice, Bob and Charlie,   are  elements of $\mathcal{A}=\{0, 1\}$. See \autoref{fig:game}. The players win if $r\vee s\vee t = a\oplus b\oplus c$, where $\vee$ and $\oplus$ respectively denote disjunction and addition mod 2, i.e for question sets $rst=000, 011, 101$ and $110$ the answer set $abc$ must satisfy $a\oplus b\oplus c = 0,1,1$ and $1$, respectively.

\begin{figure}[h!]
  \centering
    \includegraphics[width=\textwidth]{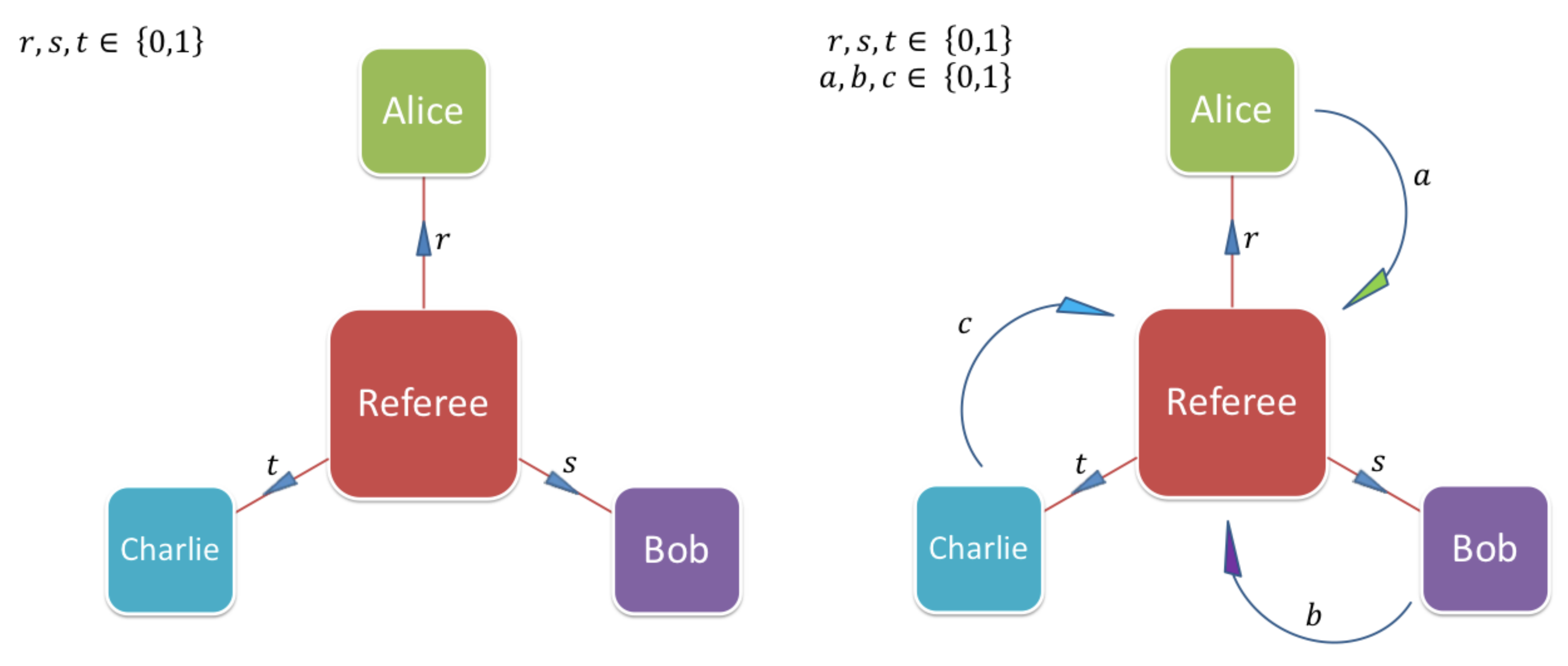}
     \caption{Three player non-local game.\label{fig:game}}
\end{figure}

In the quantum version, Alice, Bob and Charlie each possess a qubit, which they may manipulate locally. The 3-qubit   state may be entangled and used as a resource to help the players win. However, before examining how this works  let us consider how well the players can do classically, i.e. unassisted by entanglement.

A classical strategy amounts to specifying three functions, $a, b, c$, one for each player,  from the question set $\mathcal{Q}$ to the answer set $\mathcal{A}$, 

 The condition that the players win may then be written as,
\be
a(0)\oplus b(0)\oplus c(0) = 0,\;
a(1)\oplus b(1)\oplus c(0) = 1,\;
a(1)\oplus b(0)\oplus c(1) = 1,\;
a(0)\oplus b(1)\oplus c(1) = 1.\;
\ee
This implies that the best one can do is win $75\%$ of the time; the four equations cannot be simultaneously satisfied as can be seen by adding them mod 2, which yields the contradiction $0=1$ \cite{Watrous:2006}. On the other hand, the simple strategy that ``everyone always answers $1$'' satisfies three of the four equations so that the $75\%$ upper bound is actually met.

Can this be bettered when equipped with an entangled resource? The answer is a resounding yes: by sharing the GHZ state,
\be\label{eq:winstate}
\ket{\Psi}=\frac{1}{{2}}\left(\ket{000}-\ket{011}-\ket{101}-\ket{110}\right),
\ee
which is equivalent to \eqref{eq:ghz} under a local unitary rotation, they can always win \cite{Watrous:2006}.  

The winning quantum strategy is remarkably simple.   If a player receives the question ``0''  they measure their qubit in the computational basis $\{\ket{0}, \ket{1}\}$. If a player receives the question ``1''  they measure their qubit in the Hadamard basis $\{\left(\ket{0}+\ket{1}\right)/\sqrt{2}, \left(\ket{0}-\ket{1}\right)/\sqrt{2}\}$.  Their measurement outcomes are sent back as their answers. By symmetry we need only consider the two cases  $rst=000$ and $rst = 011$. 
(1)  $rst=000$: All players measure in the computational basis. From \eqref{eq:winstate} it is clear that only an odd number of 0's can appear $\Rightarrow$ $a\oplus b \oplus c=0$. Always win.
(2) $rst=011$: Alice measures in the computational basis, while Bob and Charlie measure in the Hadamard  basis. Consulting the locally rotated state,
\[
\mathds{1}\otimes H\otimes H \ket{\Psi}=\frac{1}{{2}}\left(\ket{001}+\ket{010}-\ket{100}+\ket{111}\right),
\]
where  $H$ is the unitary operator relating the computational basis to the Hadamard basis, it is clear that only an even number of 0's can appear
$\Rightarrow$ $a\oplus b \oplus c=1$. Always win.
Hence, using the GHZ entangled resource \eqref{eq:winstate}, Alice, Bob and Charlie can win 100\% of the time, outdoing the best classical strategy by 25\%. 

 One might naively expect to be able to devise a winning strategy when Alice, Bob and Charlie are equipped with a shared $W$-state, given it is also totally entangled. However, the $W$-state cannot be used win with certainty even when allowing Alice, Bob and Charlie to each measure in any pair of bases they wish \cite{Borsten:2013rca}.  Hence, despite both being totally entangled, the GHZ-state is ``more'' entangled than the $W$-state,  in this context at least. Certainly, they constitute physically inequivalent forms of entanglement. We remark that the non-local properties of the
W and GHZ states may also be compared using the sheaf-
theoretic framework of \cite{abramsky2011sheaf, PhysRevA.85.062114}.

We must therefore conclude that it is not enough to simply say a state is totally entangled - one must also specify in what way it is totally entangled. This is achieved using the paradigm of \emph{local operations and classical communication} (LOCC). See, for example, \cite{Plenio:2007}.  Roughly, given a composite quantum system we allow purely \emph{local} quantum operations  to be performed on the individual components.  These local operations may be supplemented by classical communication: the separated experimenters may communicate  via a classical channel, e-mail for example. Any number of  LOCC rounds may be performed, which makes the class of allowed operations  difficult to characterise. For an in-depth account of LOCC protocols, see \cite{Horodecki:2007}. In this manner arbitrary classical correlations between the constituent subsystems may be generated. However, no quantum correlations  may be established - all information exchanged was classical.  LOCC protocols cannot increase the amount of entanglement. 

 This motivates the concept of \emph{Stocastic} LOCC  equivalence, introduced in \cite{Bennett:1999, Dur:2000}: two states  lie in the same SLOCC-equivalence class if and only if they may be transformed into one another with some \emph{non-zero probability} using LOCC operations.
 The crucial observation is that since LOCC cannot create entanglement any two SLOCC-equivalent states must necessarily possess the same entanglement, irrespective of the particular measure used. It is this property which makes the SLOCC paradigm so suited to the task of classifying entanglement. 
 
We restrict our attention to pure states. For an $n$-qubit system, two states, 
\be
\ket{\Psi}=a_{A_1\ldots A_n}\ket{A_1\ldots A_n}, \qquad \ket{\Phi}=b_{B_1\ldots B_n}\ket{B_1\ldots B_n}, 
\ee 
are SLOCC-equivalent if and only if they are related by an element of
$\SL_1(2, \mathds{C})\times \SL_2(2, \mathds{C})\times \ldots\SL_n(2, \mathds{C})$ \cite{Dur:2000}, which will be referred to as the \emph{SLOCC-equivalence group}. The Hilbert space is partitioned into equivalence classes or orbits under the SLOCC-equivalence group. For the $n$-qubit system the space of SLOCC-equivalence classes  is given by,
\begin{equation}\label{eq:SLOCCeqiv}
\frac{{\mathds C}^2\otimes{\mathds C}^2\ldots\otimes\C^2}{\SL_1(2, \mathds{C})\times \SL_2(2, \mathds{C})\times \ldots\SL_n(2, \mathds{C})}.
\end{equation}
This  is the space of physically distinct form of entanglement. Hence, the SLOCC entanglement classification amounts to understanding \eqref{eq:SLOCCeqiv}.

\section{The 3-qubit Jordan algebra and Freudenthal triple system}\label{sec:3qubitfts}

\subsection{Cubic Jordan algebras}\label{sec:J}

A Jordan algebra $\mathfrak{J}$ is vector space defined over a ground field $\mathds{F}$ equipped with a bilinear product satisfying
\begin{equation}\label{eq:Jid}
X\circ Y =Y\circ X,\quad X^2\circ (X\circ Y)=X\circ (X^2\circ Y), \quad\forall\ X, Y \in \mathfrak{J}.
\end{equation}
The class of \emph{cubic} Jordan algebras are constructed as follows \cite{McCrimmon:2004}. Let $V$ be a vector space equipped with a cubic norm, which is. a homogeneous map of degree three, $N:V\to \mathds{F}$ s.t. $N(\lambda X)=\lambda^3N(X)$, $\forall \lambda \in \mathds{F}, X\in V$, such that 
\begin{equation}
N(X, Y, Z):=
\frac{1}{6}[N(X+ Y+ Z)-N(X+Y)-N(X+ Z)-N(Y+
Z)+N(X)+N(Y)+N(Z)]
\end{equation}
is trilinear. 
If $V$ further contains a base point $N(c)=1, c\in V$  one may define the following three maps,
    \begin{equation}\label{eq:cubicdefs}
    \begin{split}
    \Tr:V\to\mathds{F};\quad& X    \mapsto3N(c, c, X),\\
 		 S: V\times V\to\mathds{F};\quad& (X,Y)       \mapsto6N(X, Y, c),\\
    \Tr:V\times V\to\mathds{F};\quad& (X,Y)      \mapsto\Tr(X)\Tr(Y)-S(X, Y).
    \end{split}
    \end{equation}

A cubic Jordan algebra $\mathfrak{J}$, with multiplicative identity $\mathds{1}=c$, may be derived from any such vector space if $N$ is \emph{Jordan cubic}. That is: (i) the trace bilinear form \eqref{eq:cubicdefs} is non-degenerate (ii) the quadratic adjoint map, $\sharp\colon\mathfrak{J}\to\mathfrak{J}$, uniquely defined by $\Tr(X^\sharp, Y) = 3N(X, X, Y)$, satisfies
$(X^{\sharp})^\sharp=N(X)X$, $\forall X\in \mathfrak{J}$. The Jordan product is then given by 
\be
X\circ Y = \half\big(X\times Y+\Tr(X)Y+\Tr(Y)X-S(X, Y)\mathds{1}\big),
\ee
where  $X\times Y := (X+Y)^\sharp-X^\sharp-Y^\sharp$.

  \begin{definition}[Structure group $\Str(\mathfrak{J})$]  Invertible $\mathds{F}$-linear transformations $\sigma$ preserving the cubic norm up to a fixed scalar factor,
   \be
\Str(\J):=\{\sigma\in \Iso_{\F}(\J)| N(\sigma A)=\lambda N(A), \lambda\in\F\}.
\ee
The reduced structure group $\Str_0(\mathfrak{J})\subset \Str(\mathfrak{J})$ preserves the norm exactly.
  \end{definition}
  The exceptional Jordan algebra of $%
3\times 3$ Hermitian octonionic matrices, denoted $\J^{\mathds{O}}_{3}$, is perhaps the most important and well known example. The reduced structure group $\Str_0(\J^{\mathds{O}}_{3})$ in this case is given by the 78-dimensional exceptional Lie group $E_{6(-26)}$.

  Any cubic Jordan algebra element may be assigned a $\Str(\J)$-invariant \emph{rank} \cite{Jacobson:1961}. The ranks partition the  $V$.

\begin{definition}[Cubic Jordan algebra rank] A non-zero element $A\in\J$ has a rank given by:
\be\label{eq:Jranks}
\begin{split}
\textrm{\emph{Rank}} A =1& \Leftrightarrow A^\sharp=0;\\
\textrm{\emph{Rank}} A =2& \Leftrightarrow N(A)=0,\;A^\sharp\not=0;\\
\textrm{\emph{Rank}} A =3& \Leftrightarrow N(A)\not=0.\\
\end{split}
\ee
\end{definition}
\subsection{The Freudenthal triple system}\label{sec:FTS}
In 1954 Freudenthal \cite{Freudenthal:1954} found that the
133-dimensional exceptional Lie group $E_{7}$ could be understood in terms
of the automorphisms of a construction based on the  56-dimensional $%
E_{7}$-module  built from the exceptional Jordan algebra of $%
3\times 3$ Hermitian octonionic matrices.  A key feature of this construction is the triple product, hence the name. 

Given a cubic Jordan algebra $\mathfrak{J}$ defined over a field $\mathds{F}$, one is able to construct an FTS by defining the vector space $\mathfrak{F(J)}$,
\begin{equation}
\mathfrak{F(J)}=\mathds{F\oplus F}\oplus \mathfrak{J\oplus J}.
\end{equation}
An arbitrary element $x\in \mathfrak{F(J)}$ is conventionally written as a ``$2\times 2$ matrix'',
\begin{equation}
x=\begin{pmatrix}\alpha&A\\B&\beta\end{pmatrix} \quad\text{where} ~\alpha, \beta\in\mathds{F}\quad\text{and}\quad A, B\in \mathfrak{J},
\end{equation}
but for notational convenience we will often also write $x=(\alpha, \beta, A, B)$.

The FTS comes equipped with a non-degenerate bilinear antisymmetric quadratic form, a quartic form and a trilinear triple product \cite{Freudenthal:1954,Brown:1969}:
\begin{subequations}
\begin{enumerate}
\item Quadratic form $ \{x, y\}$: $\mathfrak{F}\times \mathfrak{F}\to \mathds{F}$
    \begin{equation}\label{eq:bilinearform}
    \{x, y\}=\alpha\delta-\beta\gamma+\Tr(A,D)-\Tr(B,C),\quad
\ee    
where $ x=(\alpha, \beta, A, B),  y=(\gamma, \delta, C, D)$.
\item Quartic form $q:\mathfrak{F}\to \mathds{F}$
    \begin{equation}\label{eq:quarticnorm}
        q(x)=-2[\alpha\beta-\Tr(A, B)]^2 
         -8[\alpha N(A)+\beta N(B)-\Tr(A^\sharp, B^\sharp)].
        \end{equation}
\item Triple product $T:\mathfrak{F\times F\times F\to F}$ which is uniquely defined by
	\begin{equation}\label{eq:tripleproduct}
	\{T(x, y, w), z\}=q(x, y, w, z)
	\end{equation}
	where $q(x, y, w, z)$ is the full linearisation of $q(x)$ such that $q(x, x, x, x)=q(x)$.
\end{enumerate}
\end{subequations}
\begin{definition}[The automorphism group $\Aut(\FTS)$ \cite{Brown:1969}] Invertible  $\F$-linear transformations preserving the quartic and quadratic forms:
\be
\Aut(\FTS):=\{\sigma\in\Iso_\F(\FTS)|\{\sigma x, \sigma y\}=\{x, y\},\;q(\sigma x)=q(x)\}\label{eq:brownfts}\ee
\end{definition}
\begin{lemma}[Brown \cite{Brown:1969}] The following transformations generate elements of $\Aut(\FTS)$:
\be\label{eq:ftstrans}
\begin{split}
\varphi(C):\pmtwo{\alpha}{A}{B}{\beta}&\mapsto \pmtwo{\alpha+\Tr(B,C)+\Tr(A,C^\sharp)+\beta N(C)}{A+\beta C}{B+A\times C +\beta C^\sharp}{\beta};\\
\psi(D):\pmtwo{\alpha}{A}{B}{\beta}&\mapsto \pmtwo{\alpha}{A+B\times D +\alpha D^\sharp}{B+\alpha D}{\beta+\Tr(A,D)+\Tr(B,D^\sharp)+\alpha N(D)};\\
\widehat{\tau}:\pmtwo{\alpha}{A}{B}{\beta}&\mapsto \pmtwo{\lambda\alpha}{\tau A}{^t\tau^{-1} B}{\lambda^{-1}\beta};
\end{split}
\ee
where $C,D\in\J$ and $\tau\in \Str(\J)$ s.t. $N(\tau A)=\lambda N(A)$. For convenience we also define  $\mathcal{Z}=\phi(-\mathds{1})\psi(\mathds{1})\phi(-\mathds{1})$,
\be
\mathcal{Z}:(\alpha, \beta, A, B) \mapsto (-\beta, \alpha, -B, A).
\ee
\end{lemma}

The archetypal example is given by setting $\J=\J^{\mathds{O}}_{3}$, in which case elements of $\FTS(\J^{\mathds{O}}_{3})$ transform as the $\rep{56}$ of $E_7$. Decomposing under reduced structure group $\Str_0(\J^{\mathds{O}}_{3})=E_6$ gives the branching  $\rep{56}\rightarrow\rep{1}+\rep{1}+\rep{27}+\rep{27}'$ (neglecting the $\SO(2)$ weights), where $\alpha, \beta$ comprise the singlets and $A,B$ transform as the $\rep{27}, \rep{27}'$.

The conventional concept of matrix rank may be generalised to Freudenthal triple systems in a natural and $\Aut(\mathfrak{F})$-invariant manner. 
\begin{definition}[The FTS Rank \cite{Ferrar:1972, Krutelevich:2004}] 
The rank of a non-zero element $x\in\mathfrak{F}$ is  defined by:
\be\label{eq:FTSrank}
\begin{split}
\textrm{\emph{Rank}} x=1&\Leftrightarrow \Upsilon_x(y)=0\;\forall y;\\
\textrm{\emph{Rank}} x=2&\Leftrightarrow \exists y\;\textrm{s.t.}\;\Upsilon_x(y)\not=0,\;T(x,x,x)=0;\\
\textrm{\emph{Rank}} x=3&\Leftrightarrow T(x,x,x)\not=0,\;q(x)=0;\\
\textrm{\emph{Rank}} x=4&\Leftrightarrow q(x)\not=0.\\
\end{split}
\ee
where $\Upsilon_x(y):=3T(x,x,y)+x\{x,y\}x$.
\end{definition}

\subsection{The 3-qubit  Freudenthal triple system}\label{sec:3quibitFTS}

\begin{definition}[3-qubit cubic Jordan algebra] We define the 3-qubit cubic Jordan algebra, denoted as $\J_{ABC}$, as the complex vector space $\C\oplus\C\oplus\C$ with elements  $A=(A_1, A_2, A_3)$ {and cubic norm} $N(A)=A_1A_2A_3.$
\end{definition}
Using the cubic Jordan algebra construction \eqref{eq:cubicdefs}, one finds
\begin{equation}
\Tr(A,B)= A_1B_1+A_2B_2+A_3B_3,
\end{equation}
so that, using $\Tr(A^\sharp,B)=3N(A,A,B)$, the quadratic adjoint is given by
\be
A^\sharp=(A_2A_3, A_1A_3, A_1A_2),
\ee
and therefore $(A^\sharp)^\sharp=(A_1A_2A_3A_1,A_1A_2A_3A_2,A_1A_2A_3A_3)=N(A)A.$
It is not hard to check $\Tr(A,B)$ is non-degenerate and so $N$ is Jordan cubic. In fact, it is the smallest degree three spin-factor Jordan algebra, see for example \cite{McCrimmon:2004, Borsten:2011nq}. The Jordan product is given by
\begin{equation}
A\circ B=(A_1B_1, A_2B_2, A_3B_3).
\end{equation}
The structure and reduced structure groups are given by $[\SO(2, \mathds{C})]^3$ and $[\SO(2, \mathds{C})]^2$ respectively.

\begin{definition}[3-qubit Freudenthal triple system] We define the 3-qubit Freudenthal triple system, denoted $\FTS_{ABC}$, as the complex vector space,
\be
\FTS_{ABC}:=\C\oplus\C\oplus\J_{ABC}\oplus\J_{ABC}.
\ee 
\end{definition}
We identify the  eight  components of $\FTS_{ABC}$ with the eight three qubit basis vectors $\ket{ABC}$  so that, for $\ket{\psi}=a_{ABC}\ket{ABC}$,
\begin{equation}
        \begin{pmatrix} \alpha & (A_1, A_2, A_3) \\
                        (B_1,B_2,B_3)   &  \beta
        \end{pmatrix}\leftrightarrow \Psi:=\
        \begin{pmatrix} a_{111} & (a_{001}, a_{010}, a_{100}) \\
                       (a_{110}, a_{101}, a_{011})&a_{000}
        \end{pmatrix}.
\end{equation}

Crucially, the automorphism group of $\FTS_{ABC}$, as defined in \eqref{eq:brownfts}, is $\SL(2, \C)\times\SL(2, \C)\times\SL(2, \C)\rtimes S_3$, where $S_3$ is the three-qubit permutation group.  On including the permutation group the three biseparable entanglement classes, $A$-$BC$, $B$-$CA$ and  $C$-$AB$, are identified.

\section{FTS entanglement classification}\label{sec:FTSclass}

The ranks of the FTS are in fact the entanglement classes: all states of a given rank $r=1,2,3$ are SLOCC-equivalent. Rank four states constitute a $\dim_\C=1$ family of equivalent states parametrised by $q(\Psi)$. More specifically, we have:  (Rank 1) Totally separable states $A$-$B$-$C$, (Rank 2) Biseparable states $A$-$BC$, (Rank 3) Totally entangled  $W$ states, (Rank 4) Totally entangled GHZ states. The rank 4 GHZ class is regarded as maximally entangled in the sense that it has non-vanishing quartic norm. See \autoref{tab:merge}. To prove this statement we use the following  result, which is an extension of  lemma 24 in \cite{Krutelevich:2004}:

\begin{table}
\begin{tabular*}{\textwidth}{@{\extracolsep{\fill}}cccccM{c}M{c}cc}
\toprule
& \multirow{2}{*}{Class}  &\multirow{2}{*}{Rank} &\multirow{2}{*}{Representative state}& & \multicolumn{2}{c}{\text{FTS rank condition}}               & & \\
\cline{5-8}
&                        &   &                    & & \text{vanishing}                     & \text{non-vanishing} & & \\
\hline
& Null                   & 0  	&	 0            & & \Psi  				                  & -                    & & \\
& $A$-$B$-$C$            & 1  &	$\ket{111}$	             & & \Upsilon_\Psi(\Phi)& \Psi 		         & & \\
& $A$-$BC$               & 2 &	$\ket{111}+\ket{100}$	             & & T(\Psi,\Psi,\Psi) 					  &\Upsilon_\Psi(\Phi)              & & \\
& W                      & 3  		   &        $\ket{111}+\ket{001}+\ket{010}$  & & q(\Psi)							  & T(\Psi,\Psi,\Psi) 	 & & \\
& GHZ                    & 4  		&      $\ket{111}+\ket{001}+\ket{010}+k\ket{100}$       & & -                                    & q(\Psi)              & &\\
\hline
\end{tabular*}
\caption{The entanglement classification of three qubits as according to the FTS rank system.\label{tab:merge}}
\end{table}

\begin{lemma}\label{lem:rcf}
Every state $\Psi$ is SLOCC-equivalent to the reduced canonical form:
\be\label{eq:rcf}
\Psi_{\text{\emph{red}}}=(1,  0, A, 0)\longleftrightarrow\ket{111}+a_{001}\ket{001}+a_{010}\ket{010}+a_{100}\ket{100}
\ee
where $\Det a=4a_{011}a_{101}a_{110}$.
\end{lemma}
\begin{proof} We start with a generic state $(\alpha, \beta, A, B)$. We may always assume $B$ is non-zero. (If $B=0$ and $A\not=0$ use $\mathcal{Z}$. If $B=0, A=0$ then we may assume that  $\alpha$ non-zero, using $\mathcal{Z}$ if necessary, and apply $\psi(D)$ to get a non-zero $B$-componant, as required.)  Apply $\phi(C)$ with $C^{\sharp}=0$ s.t. and  $\alpha\mapsto \alpha +\Tr(C, B)=1$, which is  always possible since $B\not=0$,  the trace form is non-degenerate and $\J_{ABC}$ is spanned by its rank 1 elements. We are left with a new state $(1, \beta', A', B')$. Now apply $\psi(D)$ with $D=-B'$, so that $(1, \beta', A', B')\mapsto (1, \beta'', A'', 0)$.
 
Hence we may  assume from the outset that our state is in the reduced form
$
(1, \beta, A, 0)$. We now  show that we may also set $\beta=0$.  Since the Jordan ranks \eqref{eq:Jranks} partition $\mathfrak{J}_{ABC}$, there are  four subcases to consider: (i) $A=0$, (ii) $A\not=0, A^\sharp=0$,   (iii) $A^\sharp\not=0, N(A)=0$ and (iv) $N(A)\not=0$.

(i) Apply $\phi(C)$ with $C\not=0, C^\sharp=0$: $(1, \beta, 0, 0)\mapsto (1, \beta, \beta C, 0)$ so that we are now in case (ii).

(ii) We may assume with out loss of generality, $A=(a, 0, 0)$, $a\not=0$. Let $C=(0,c,0)$,
\[
\phi(C):(1, \beta  , A, 0)\mapsto(1, \beta  , A', A\times C),
\]
where $A'=(a, \beta c, 0) \Rightarrow A'^\sharp\not=0, N(A')=0$.
Apply $\psi(D)$ with $D=-A\times C$, \[
\psi(D):(1, \beta  , A', A\times C)\mapsto(1, \beta, A', 0),
\]
so that  we find ourselves in case (iii).

(iii) Without loss of generality we may assume $A=(0, a_2, a_3)$, $a_2,a_3\not=0$. Let $C=(c,0,0)$,
\[
\phi(C):(1, \beta  , A, 0)\mapsto(1, \beta  , A', B'),
\]
where $A'=(\beta c, a_2, a_3)$ and  $B'=(0 , c a_3, c a_2)$. Apply $\psi(D)$ with $D=-B'$, \[
\psi(D):(1, \beta  , A', B')\mapsto(1, \beta-2ca_2a_3, A'', 0),
\]
and let $c=\beta/(2a_2a_3)$ to obtain the canonical reduced  form $(1,0,A'',0)$.

(iv) The augment of case (iii) does not require $N(A)=0$ and, hence, also applies to the present case.
\end{proof}

Given \autoref{lem:rcf}, the entanglement classifications  follows almost trivially. The the rank conditions applied to the reduced canonical form \eqref{eq:rcf} dramatically simplify and imply the that each rank (up to permutation) is represented a state corresponding to the classes of \autoref{tab:merge}:
\be\label{eq:states}
\begin{array}{lllllllll}
\textrm{Rank} \Psi_{\text{red}}  =1& \Leftrightarrow &A=0,&\quad \Rightarrow &\quad\Psi_{\text{red}}= \ket{111},\\
\textrm{Rank} \Psi_{\text{red}} =2& \Leftrightarrow &A^\sharp=0,\;A\not=0,&\quad \Rightarrow&\quad\Psi_{\text{red}}= \ket{111}+\ket{001},\\
\textrm{Rank} \Psi_{\text{red}} =3& \Leftrightarrow &N(A)=0,\;A^\sharp\not=0,&\quad \Rightarrow&\quad\Psi_{\text{red}}= \ket{111}+\ket{001}+\ket{010},\\
\textrm{Rank} \Psi_{\text{red}} =4& \Leftrightarrow &N(A)\not=0,&\quad \Rightarrow&\quad\Psi_{\text{red}}= \ket{111}+\ket{001}+\ket{010}+k\ket{100},\\
\end{array}
\ee
where $k\not=0$. 
To reach the final form of the representative states on the far left of \eqref{eq:states}, we have scaled using $\hat{\tau}\in [\SO(2, \mathds{C})]^2\cong \Str_0(\J_{ABC})\subset\Aut(\FTS_{ABC})$, see \eqref{eq:ftstrans}. 

Since the ranks partition $\FTS_{ABC}$, this completes the orbit (entanglement class) classification of \eqref{eq:SLOCCeqiv}, which is summarised in \autoref{tab:merge}. As claimed, elements of rank 1, 2 and 3 belong to a single orbit corresponding  to totally separable, biseparable and $W$ states, respectively. Rank 4 elements belong to a one complex dimensional family of orbits, parametrised by $q(\Psi)=-8k$, and correspond to GHZ states. Note, the rank classification places GHZ above $W$; they are not merely inequivalent, but ordered, as reflected by three player non-local game \cite{Borsten:2013rca}.

\section{Generalising to an $n$-qubit FTS}\label{sec:nquibitFTS}

The success of the  FTS classification of 3-qubit entanglement naturally raises the question of generalisation. Are there other composite quantum systems amenable to the FTS treatment?  What about mixed states? Is there an extension to an arbitrary number of qubits?  

Remarkably, it has already been shown that a variety of the FTS based on cubic Jordan algebras  provide the pure state SLOCC entanglement  classification of composite quantum systems, including mixtures of bosonic and fermonic qudits \cite{Borsten:2008, Borsten:2008wd, levay-2008, Levay:2009, Borsten:2012fx, Levay:2012tg}. Moreover, in the case of three qubits the FTS may be used in the  mixed state classification  \cite{Szalay:2012}.  

In the subsequent sections we focus on the final question posed above. While there is no arbitrary $n$-qubit FTS \emph{per se}, we may attempt to identify those aspects of the 3-qubit FTS which naturally generalise to an arbitrary number of qubits in the hope that these universal features are illuminating. This is the approach taken here.

\subsection{The $n$-qubit state reorganised}
Recall, the Jordan algebra formulation of the FTS corresponded to decomposing the representation carried by the FTS under the $\Str_0(\J)\subset\Aut(\FTS)$.  In the case of three qubits we found the  state split into the direct sum of four pieces,
\be
\alpha=a_{111},
\beta=a_{000},
A=(a_{100},a_{010},a_{001}),
B=(a_{011},a_{101},a_{110}),
\ee
where $\alpha, \beta$ are $\Str_0(\J)$ singlets. This leads us to the first important observation: $\alpha, \beta, A, B$ are the closed subsets under the 3-qubit permutation group $S_3$. Indeed, if we are only interested in the  SLOCC entanglement classification up to permutations, as we are, it is only natural to work with  $S_n$-closed subsets as the basic building blocks. Hence,  the $2^n$ state vector coefficients will be collected into the  $n+1$  subsets closed under $S_n$. It will prove convenient to represent these subsets using $n+1$ totally symmetric tensors with ranks ranging from 0 to $n$,
\be
\mathcal{A}_n:=\{A_0, A_{i_1}, A_{i_1i_2},\ldots, A_{i_1i_2\ldots i_n}\}, \qquad\text{where}\qquad i_k=1, 2, \ldots,n,
\ee
which are vanishing on any diagonal, i.e. $A_{i_1i_2\ldots i_n}=0$ if any two indices are the same. The counting of components goes like $p$-forms, correctly yielding a total of $2^n$ independent coefficients. Indeed, we could have just as well defined $\mathcal{A}_n$ as a set of totally antisymmetric tensors of ranks 0 to $n$, avoiding the need to impose tracelessness. However, the 3-qubit FTS structure most naturally transfers over using the symmetric formulation, so we will stay with that convention here.

 For three qubits we have (with a slight abuse of notation for $A_{ijk}$),
\be
A_0=a_{000},\quad
A_i=\begin{pmatrix}a_{100}\\a_{010}\\a_{001}\end{pmatrix},\quad
A_{ij}=\pmthree{0}{a_{110}}{a_{101}}{a_{110}}{0}{a_{011}}{a_{101}}{a_{011}}{0},\quad
A_{ijk}=a_{111}.
\ee
Note, numbering the qubits from left to right, the values of the indices on the symmetric tensors determine which indices  on its corresponding state vector coefficient take the value 1. For example, $A_1=a_{100}, A_2=a_{010}, A_3=a_{001}$ and $A_{12}=a_{110}, A_{13}=a_{101}$ and so on. We are grateful to Duminda Dahanayake for pointing out  this rule. 

\subsection{The $n$-qubit algebra}

The second feature we might hope to generalise is the set of cubic Jordan algebra maps, $A\times B, \Tr(A, B)$, $N(A)$, see \autoref{sec:J}, which played such a key role in the construction of the various covariants and invariants. Recall, group theoretically these maps correspond to picking out certain irreps appearing in the tensor product of the $\Str_0(\J)$-representation carried  by $A, B\in\J$. For example, in the case $\Str_0(\J^{\Oct}_3)=E_{6(-26)}$, with $A$ and $B$ transforming as the $\rep{27}$, $A\times B$ is the $\rep{27}'$ in $\rep{27}\times\rep{27}=\rep{27}'_s+\rep{351}_a+\rep{351'}_s$ and  $N(A)$ is the singlet in $\rep{27}\times\rep{27}\times\rep{27}$. Each of the cubic Jordan algebra maps may be written using the irreducible $E_{6(-26)}$ invariant tensors, $d_{ijk}$ and $d^{ijk}$, where a downstairs (upstairs) $i=1,2, \ldots, 27$ transforms as a $\rep{27}$ ($\rep{27}'$).  For example,  
$
(A^\sharp)^i=\frac{1}{2!}d^{ijk}A_jA_k,$ and $N(A)=\frac{1}{3!}d^{ijk}A_iA_jA_k$.
For the sake of clarity, we will often drop the combinatorial factors in the following. For three qubits, the invariant tensors were simply
\be
d_{ijk}=|\epsilon_{ijk}|,\quad d^{ijk}=|\epsilon^{ijk}|,
\ee
which naturally suggests the $n$-qubit generalisation,
\be
d_{i_1\ldots i_n}:=|\epsilon_{i_1\ldots i_n}|,\quad d^{i_1\ldots i_n}:=|\epsilon^{i_1\ldots i_n}|.
\ee
This allows us to dualise a rank $p$ tensor, 
\be
A^{i_1i_2\ldots i_{n-p}}:=\frac{1}{p!}d^{i_1i_2\ldots i_{n-p}i_{n-p+1}\ldots i_{n}}A_{i_{n-p+1}\ldots i_{n}}.
\ee
For an $n$-qubit state, rank $p$ pairs $A_{i_1i_2\ldots i_p}, A^{i_1i_2\ldots i_p}$ are precisely bit-flip related.  For example, for three qubits, $A_i=(a_{100}, a_{010}, a_{001})$ and $A^i=(a_{011}, a_{101}, a_{110})$. This is crucial for building $S_n$ invariants.
Equipped with $d_{i_1\ldots i_n}, d^{i_1\ldots i_n}$ the $n$-qubit space $\mathcal{A}_n$ of symmetric tensors may be endowed with a pseudo-algebraic structure:  $\mathcal{A}_n$  is closed so long as we compose the tensors by contracting with $d_{i_1\ldots i_n}$ and $d^{i_1\ldots i_n}$.  

\subsection{The $n$-qubit generalised FTS transformations}

 The FTS transformations \eqref{eq:ftstrans} for three qubits in the our new notation are given by,
\be\label{eq:3qtrans}
\begin{split}
\phi(C):\begin{pmatrix}A^{ijk}\\A^{ij}\\A^i\\A^{0}\end{pmatrix}&\mapsto\left(
\begin{array}{lclclcl}
                A^{ijk}\\
         A^{ij} &+&                C_k  A^{ijk}\\
A^i &+&C_jA^{ij}&+&C_jC_kA^{ijk}\\
 A^{0}    &+&C_iA^i   &+&  C_{i}C_j  A^{ij}    &\;\,+ &   C_{i}C_jC_kA^{ijk}
\end{array}\;\,\right),\\
\psi(D):\begin{pmatrix}A_{0}\\A_i\\A_{ij}\\A_{ijk}\end{pmatrix}&\mapsto\left(
\begin{array}{lclclcl}
 A_{0}  &+&    D^iA_i   &+&     D^i D^{j}A_{ij} &+& D^{i}D^jD^kA_{ijk}  \\
 A_{i}  &+&      D^jA_{ij}    &+&   D^{j}D^{k}   A_{ijk}    \\
A_{ij} &+& D^kA_{ijk} \\
                  A_{ijk}
\end{array}\right),\\
\hat{\tau}(\lambda):\begin{pmatrix}A_{0}\\A_i\\A^i\\A^{0}\end{pmatrix}&\mapsto\left(
\begin{array}{ll}
d_{lmn}\lambda^l\lambda^m\lambda^n &A_0\\
\xi^ld_{lmn}\lambda^m\lambda^n& A_i\\
\xi^l\xi^md_{lmn}\lambda^{n}& A^{i}\\
\xi^l\xi^m\xi^nd_{lmn}&A^{0}
\end{array}\right),
\end{split}
\ee
where  $\lambda\in\C-\{0\}$ and $\xi=\lambda^{-1}$. For $\phi(C), \psi(D)$ we have made a judicious choice of dualisations in order to make the correct $n$-qubit generalisation manifest.
Under $\phi(C)$ a rank $p$ tensor $A^{i_1i_2\ldots i_p}$ transforms into the sum of all tensors $A^{i_1i_2\ldots i_q}$ with $q\geq p$ contracted with the necessary powers of $C_i$ to give back rank $p$. Explicitly,
\be
\begin{array}{lll}
A^{i_1i_2\ldots i_{n}}&\mapsto &[A^{i_1i_2\ldots i_{n}}],\\
A^{i_1i_2\ldots i_{n-1}}&\mapsto &[A^{i_1i_2\ldots i_{n-1}}+C_{i_n} A^{i_1i_2\ldots i_{n}}],\\
A^{i_1i_2\ldots i_{n-2}}&\mapsto &[A^{i_1i_2\ldots i_{n-2}}+C_{i_{n-1}} A^{i_1i_2\ldots i_{n-1}}+C_{i_{n-1}}C_{i_{n}} A^{i_1i_2\ldots i_{n-1} i_n}],\\
\;\;\vdots&\;\vdots &\;\;\;\vdots\\
A^{0}&\mapsto& [A^{0}+C_i A^{i}+C_iC_j A^{ij}+\cdots+C_{i_1}C_{i_2}\cdots C_{i_n} A^{i_1i_2\ldots i_n}],\\
\end{array}
\ee
Similarly, under $\psi(D)$ a rank $p$ tensor $A_{i_1i_2\ldots i_p}$ transforms into the sum of all $A_{i_1i_2\ldots i_q}$ with $q\geq p$, contracted with the necessary powers of $D^i$ to give back rank $p$,
\be
\begin{array}{lll}
A_{0}&\mapsto& [A_{0}+D^i A_{i}+D^iD^j A_{ij}+\cdots+D^{i_1}D^{i_2}\cdots D^{i_n} A_{i_1i_2\ldots i_n}],\\
A_{i}&\mapsto &[A_{i}+D^j A_{ij}+\cdots+D^{i_2}D^{i_3}\cdots D^{i_n} A_{ii_2i_3\ldots i_n}],\\
A_{ij}&\mapsto &[A_{ij}+\cdots+D^{i_3}D^{i_4}\cdots D^{i_n} A_{ij i_3i_4\ldots i_n}],\\
\;\;\vdots&\;\vdots &\;\;\;\vdots\\
A_{i_1i_2\ldots i_n}&\mapsto &[A_{i_1i_2\ldots i_n}].\\
\end{array}
\ee
The generalised $\hat{\tau}(\lambda)$ may also be concisely written using this notation,
\be
\hat{\tau}(\lambda):A_{i_1i_2\ldots i_p}\mapsto  \xi^{j_1}\xi^{j_2}\cdots\xi^{j_p}d_{j_1j_2\ldots j_n}\lambda^{j_{p+1}} \cdots\lambda^{j_{n-1}}\lambda^{j_n} A_{i_1i_2\ldots i_p}.
\ee
Hence, adopting the notational convention $A_{[p]}$ ($A^{[p]}$) for a rank $p$ tensor with downstairs (upstairs) indices, the $n$-qubit transformations $\phi, \psi$ and $\hat{\tau}$ may be summarised as follows,
\be
\begin{array}{rlll}
\phi(C_{[1]}): &A^{[p]}&\mapsto &\sum_{k=p}^{n}C^{(k-p)}_{[1]}A^{[k]},\\
\psi(D^{[1]}): &A_{[p]}&\mapsto &\sum_{k=p}^{n}D^{[1](k-p)}A_{[k]},\\
\hat{\tau}(\lambda^{[1]}): &A_{[p]}&\mapsto &\xi^{[1](p)}d_{[n]}\lambda^{[1](n-p)} A_{[p]},\\
\end{array}
\ee
One useful observation that immediately follows is that one can always assume $A_0=1, A^i=0$ under SLOCC.  It is also clear that the Jordan ranks \eqref{eq:Jranks} naturally generalises to a set of rank conditions on $A_{[1]}, A^{[1]}$.

However, we have yet to develop a systematic method  for writing covariants/invariants in this scheme. One example, though, defined for $n$ qubits, is given by
\be\label{eq:quadftsinv}
(\mathcal{A}_n, \mathcal{B}_n)	=\sum_{k=0}^{n}\frac{(-1)^k}{k!}A_{[k]}B^{[k]}.
\ee
This is symmetric (antisymmetric) for even (odd) $n$. It is simply the determinant in the 2-qubit case and the antisymmetric bilinear form of the FTS in the 3-qubit case. There are four algebraically independent 4-qubit permutation and SLOCC-equivalence group  invariants \cite{Briand:2003a} of  order two, six, eight and twelve,  \eqref{eq:quadftsinv} being the order two example.

\subsubsection*{2-qubit example:}
The 2-qubit state corresponds to,
\be
\ket{\psi}\leftrightarrow\Psi=\{A_0, A_{i}, A_{ij}\}, \qquad i,j=1,\ldots , 2
\ee
where
\be
\{A_0=a_{00}, \quad A_{i}=\begin{pmatrix}a_{10}\\a_{01}\end{pmatrix},\quad A^{0}=a_{11}\}.
\ee
It is easy to verify in this scheme that every state $\mathcal{A}_2$ is SLOCC-equivalent to the reduced canonical form:
\be
\mathcal{A}_{2}^{\text{red}}=\{1,  0, k\},\longleftrightarrow\ket{00}+k\ket{11}
\ee
where $k=(\mathcal{A}_2, \mathcal{A}_2)$. First we may always assume $A_i$ is non-zero. (If $A_i=0$ then we may assume that  $A_0$ non-zero using $\mathcal{Z}$ if necessary.) Now apply $\phi(C)$ to get a non-zero $A_i$. Apply $\psi(D)$ with $d_{ij}D^iD^j=0$ so that $A_0\mapsto A_0 +D^iA_i$. Choose $D$ s.t. $A_0 +D^iA_i=1$. Finally, apply $\phi(C)$ with $C_i=-A_i$.

This gives us the well known 2-qubit SLOCC entanglement classification: there are just two classes, separable and entangled, the latter being a family of orbits parametrised by $(\mathcal{A}_2, \mathcal{A}_2)=k$.

\section*{Acknowledgements} 

I would like to extend my gratitude to the conference organisers and especially to Professor Radu Iordanescu. Many thanks to D Dahanayake, MJ Duff, H Ebrahim, and W Rubens, with whom this work was done. The work of LB is supported by an Imperial College Junior Research Fellowship.

\section*{References}

\providecommand{\newblock}{}

\end{document}